\documentclass[%
 reprint,
 onecolumn,
 amsmath,amssymb,
 aip,
]{revtex4-1}

\usepackage{xcolor}
\usepackage{graphicx}

\usepackage{amsthm}

\newtheorem{lemma}{Lemma}

\usepackage{subcaption}
\usepackage{dcolumn}
\usepackage{bm}
\usepackage{amsmath,comment}

\usepackage[T1,T2A]{fontenc}
\usepackage[utf8]{inputenc}
\usepackage[russian,english]{babel}

\newcommand{\secref}{Section~\ref}
\newcommand{\appref}{Appendix~\ref}
\newcommand{\figref}[2][{}]{\figurename~\ref{#2}#1}

\newcommand*\rmd{\mathop{}\!\mathrm{d}}
\newcommand*\rmi{\mathrm{i}}

\newcommand{\ddif}[2]{\dfrac{\rmd #1}{\rmd #2}}
\newcommand{\pdif}[2]{\dfrac{\partial #1}{\partial #2}}

\begin{document}


\title{Island formation in collisionless kinetic plasmas with perturbed orbits}
\author{B. J. Q. Woods}
	\affiliation{School of Mathematics, University of Leeds, Woodhouse Lane, Leeds, LS2 9JT, United Kingdom}
	\affiliation{Department of Physics, York Plasma Institute, University of York, Heslington, York, YO10 5DD, United Kingdom}

\date{\today}

\begin{abstract}
In the vinicity of phase-space resonances for a given species of plasma, the particle distribution function is flattened as free energy is exchanged between the plasma and resonant electromagnetic waves. Here, we present action-angle variables which explicitly separate the adiabatic invariant from the contribution that arises from perturbation of the Hamiltonian over the course of a single orbit. Then, we perform similar analysis for the orbit period, allowing one to identify a generating function $\psi$ which modifies the zeroth order contribution to the orbit period (unperturbed orbit) in the case where the particle energy is not time-invariant. We posit that a population of `quasi-trapped' particles cross the separatrix, directly allowing for island growth and decay. Then, we employ $\psi$ in the ensemble case by considering how this generating function allows for solutions of the 6+1D electrostatic Vlasov equation where finite growth and frequency sweeping of modes occur. Finally, we derive an approximate form for $\psi$ outside of the separatrix, allowing for qualitative observation of phase-space shear and island formation.
\end{abstract}

\maketitle



\section{Introduction}
\label{sec:intro}
The Boltzmann equation \cite{wesson2011tokamaks} has been widely explored in plasma physics to kinetically model plasma in linear and nonlinear regimes. The collisionless variant (the Vlasov equation) \cite{vlasov1938vibration} is widely explored due to the low collisionality sometimes found in these systems, when long-range effects dominate over small scale ballistic effects.

Nonlinear solutions related to those given by Bernstein, Greene and Kruskal (so called BGK modes) are widely explored in the literature. \cite{bernstein1957exact,chen2002BGK} Extension to the theory has been explored by Berk and Breizman, where a phase space `hole' and `clump' (relative decrease and increase respectively on the particle distribution function) leads to frequency sweeping. \cite{berk1995numerical,berk1996nonlinear,berk1997spontaneous} These solutions are of particular interest to the fusion community, as Alfv\'{e}nic frequency sweeping is strongly correlated with fast ion loss in tokamaks. \cite{heidbrink2008basic,chen2016physics,duarte2018study,todo2019introduction}

\begin{figure}[h!]
	\centering
	\includegraphics[width=0.4\textwidth]{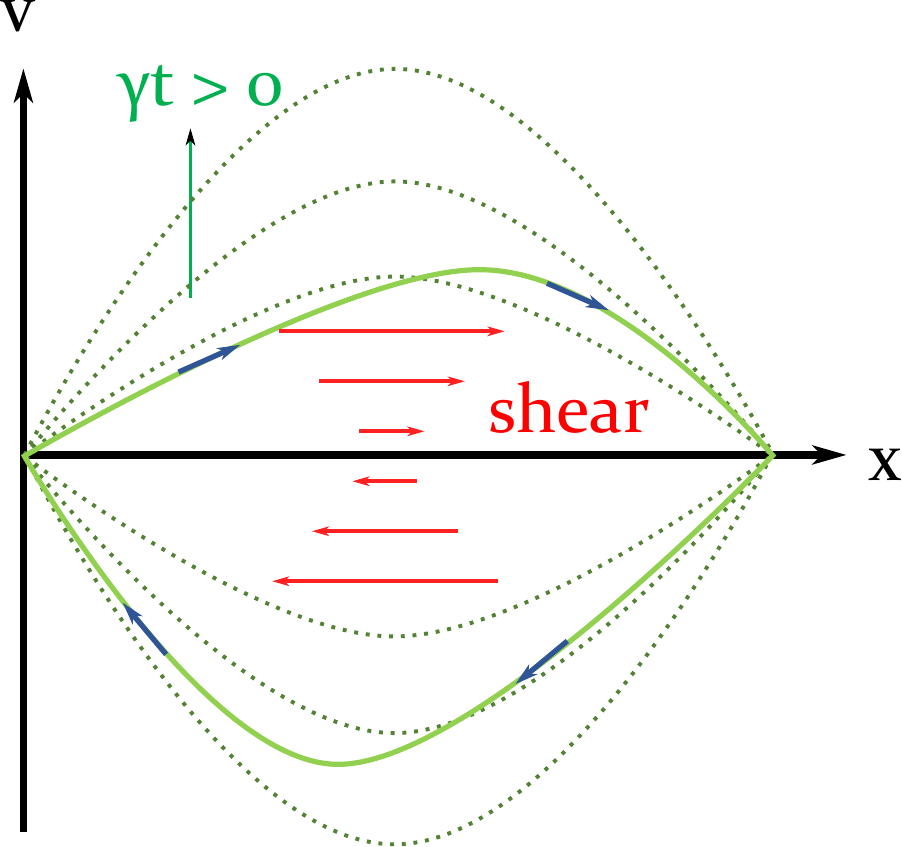}
	\caption{\textbf{Phase space shear induced by finite $\gamma$.} A passing or quasi-passing particle (see \secref{sec:period}) interacting with a wave undergoing growth or decay can be thought of as moving across a continuum of conservative orbits. The resultant orbit becomes elongated close to the X points of the orbit, resulting in phase space shear. The direction of the shear is related to the sign of the product $q\gamma$; as the sign of $q\gamma$ changes, the sign of the shear is also inverted.}
	\label{fig:intro_shear}
\end{figure}

Resonant particles exhibit closed orbits, forming phase islands which are populated as energy is exchanged with the wave, and depleted via collisions and other dissipative mechanisms. One expects that from the Vlasov equation, the spatial advection of the distribution function ($\mathbf{v} \cdot \nabla f$) contributes to phase-space shear during mode growth/decay (see \figref{fig:intro_shear}). As a result of shear, the distribution function diffuses in $\mathbf{k}$-space, a manifestation of the inherent nonlinear dynamics. BGK modes do not experience phase-space shear, and therefore the theory cannot fully capture mode growth/decay, and consequently island formation.

Berk-Breizman models and extensions thereof rely on an assumption of a `continuum of BGK modes' such that the action variables for the system are approximately adiabatic invariants. The evolution of waves in systems near marginal stability during near-constant wave amplitude has therefore been heavily explored, allowing for analytical theory of the temporal behaviour of marginally unstable modes after nonlinear saturation. \cite{lesur2010berk,degol2010nonlinear,hezaveh2017impact,dudkovskaia2019stability}

These systems are largely examined computationally, with work including phase-space island formation \cite{lilley2014formation}, stochastic phase-space island destruction \cite{woods2018stochastic}, and phase-space island evolution in tokamak geometries \cite{lang2011nonlinear,wang2013radial,meng2018resonance}. Further theory describing the growth and decay of these structures may allow for better understanding of how to modify the lifetime of phase-space islands, as well as how to manipulate holes and clumps from an `excitation and recombination' point of view, analogous to hole-electron pairs as widely explored in solid state physics. \cite{kittel1953introduction} This may in turn allow for better understanding of mode avalanching in tokamaks, which feasibly could be caused by spontaneous creation of large populations of holes or clumps (analogous to population inversion in lasers).

In \secref{sec:action}, we present action-angle variables which are explicitly represented as an adiabatic invariant, and a finite perturbation that arises due to variation of the action-angle variables during the course of an orbit. We then show that `adiabaticity' as commonly described in the literature \cite{berk1997spontaneous} is formally given by the case where the finite perturbations vary slower than deformations to the Hamiltonian.

In \secref{sec:period}, we show that by considering single particle orbits under the Lorentz force, one naturally obtains a first-order nonlinear ordinary differential equation (ODE). In previous work, we examined solutions to the Boltzmann equation in single species plasma where the 1D particle distribution function is given by a functional $f = f[\epsilon]$, where $\epsilon(x,v,t)$ is for a single wave, equal to the single particle energy at a resonance $v = v_{\textrm{res.}}$.\cite{woods2019analytical} Here, we show that the first-order nonlinear ODE can be reformulated instead as defining an energy-like quantity $\epsilon$, given by the sum of the single particle energy and a generating function $\psi$. We show that the orbit period $\tau$ can be represented as an adiabatic part (which one can identify as the unperturbed orbit), and a perturbation. In the case where the perturbation is small, $\tau$ is approximately given by the zeroth order contribution, which we refer to as a `near-conservative orbit'. We therefore posit that the generating function $\psi$ may allow for a subpopulation of particles to cross the separatrix, leading to island formation.

In \secref{sec:generating} we use a nonlinear Laplace decomposition of the electric field to derive a form for $\psi$ that approximately solves the Vlasov equation outside of island separatrices. Finally, in \secref{subsec:generating_nowave}, we show that the resultant phase-space contours share similarities with the particle orbits found in \secref{sec:period}, additionally exhibiting phase-space shear as expected during mode growth/decay.

\section{Action-angle variables}
\label{sec:action}

Action-angle variables allow one to deduce the form of the Hamiltonian in systems where the Hamiltonian is slowly evolving, and are commonly utilised in tokamak physics \cite{white2014theory,duarte2017quasilinear}. Here, we will examine the following canonical transformation of the Hamiltonian $H(\mathbf{q},\mathbf{p},t)$ which is enabled by a type-2 generating function $Q_2(\mathbf{q},\mathbf{J},t)$ \cite{goldstein1950classical}:

\begin{subequations}
\begin{align}
w_k &= \pdif{Q_2}{J_k} \\
p_k &= \pdif{Q_2}{q_k} \\
K(\mathbf{w},\mathbf{J},t) &= H(\mathbf{q},\mathbf{p},t) + \pdif{Q_2}{t}
\end{align}
\end{subequations}

where $\mathbf{J} = \{J_k\}$ are the action variables in the system, and $\mathbf{w} = \{w_k\}$ are the angle variables. The new Hamilton's equations are given by:

\begin{equation}
\label{eq:action_hamiltons}
\ddif{}{t}g(\mathbf{w},\mathbf{J},t) = \{g, K\} + \pdif{g}{t}
\end{equation}

where $\{g,K\}$ denotes the Poisson bracket in the new canonical phase space:

\[
\{g,K\} := \sum\limits_k \pdif{g}{w_k} \pdif{K}{J_k} -  \pdif{K}{w_k} \pdif{g}{J_k}
\]

Here, we define the following new generalized momenta $\{J_k\}$:

\begin{equation}
J_k := \delta J_k(\mathbf{q},t) + \oint\limits_{H(q_k,p_k,t) = \textrm{const.}} p_k \rmd q_k
\end{equation}

where the integral is performed over a closed phase space orbit at time $t$ (such that it is the unperturbed orbit), and the new generalized momenta are constants of motion. Our definition of the action variables differs here from the literature by explicitly separating the unperturbed orbit and the perturbation to the action. As such, the last term on the right hand side represents an adiabatic invariant, while the first term on the right hand side represents deviation from adiabaticity. 

One can further demand that the new Hamiltonian is not explicitly a function of time. To achieve this, one must use a generating function such that:

\begin{equation}
\label{eq:generating_restriction}
\pdif{}{t} \left( \pdif{Q_2}{t} + H \right) = 0
\end{equation}

From the new Hamilton's equations for $\{w_k\}$ and $\{J_k\}$ given by \eqref{eq:action_hamiltons}, if $\delta J_k$ is small and $J_k$ is approximately constant over one period:

\begin{align*}
w_k = \int\limits_{0}^{t} \pdif{}{J_k}K(\mathbf{J}) \rmd t + \textrm{const.} &;& 0 \approx - \pdif{K}{w_k}
\end{align*}

where one should note that $K$ cannot be a function of $\mathbf{w}$. The canonical transformation is therefore such that the new Hamiltonian is solely a function of the constants of motion. Each one is representable by a quantity which resembles classical action, plus a perturbation. If one integrates over a full period $\tau_{(0)}(\mathbf{J})$ of the orbit at time $t$:

\[
\Delta w_k \approx \pdif{}{J_k}K(\mathbf{J}) \tau_{(0)} (\mathbf{J},t)
\]

But one can also represent the variation of the angle variable $w_k$ over one orbit by:

\[
\Delta w_k = \delta w_k(\mathbf{J},t) + \oint\limits_{H(q_k,p_k,t) = \textrm{const.}} \pdif{w_k}{q_k} \rmd q_k
\]

such that $\Delta w_k$ is constructed from the unperturbed orbit contribution (where $C : H(q_k, p_k, t) = \textrm{const.}$), plus the variation arising from the orbit changing as a function of time. Then:

\[
\Delta w_k = 1 + \delta w_k
\]

and accordingly:

\[
\pdif{}{J_k}K(\mathbf{J}) \approx \dfrac{1 + \delta w_k(\mathbf{J},t)}{\tau_{(0)} (\mathbf{J},t)}
\]

such that the left hand side is equal to the bounce frequency $(1 / \tau_{(0)})$ plus a contribution arising from temporal perturbations to the particle orbits. The angle variables are therefore given by:

\[
w_k \approx \dfrac{1 + \delta w_k(\mathbf{J},t)}{\tau_{(0)} (\mathbf{J},t)} t + \textrm{const.}
\]

Altogether, this is simply a representation of Noether's theorem. The continuous symmetries here are that $K$ is invariant under a temporal transformation or a translation in $w_k$. The former leads to conservation of energy, and the latter manifests with $J_k$ as constants of motion. One finds that the Hamiltonian can be represented in the form:

\[
H(\mathbf{q},\mathbf{p},t) = K(\mathbf{J}) + \delta H(\mathbf{q},\mathbf{J},t)
\]

such that $\delta H := -\partial Q_2 / \partial t$ represents a explicitly time-varying perturbation to the Hamiltonian. Here, we formally define adiabatic invariants as:

\begin{equation}
\boxed{
\textrm{adiabatic} \, : \, \pdif{}{t} \ln \left(J_k - \delta J_k\right) \ll \pdif{}{t} \ln H
}
\end{equation}

such that the logarithmic rate of change of $(J_k - \delta J_k)$ is small compared to the Hamiltonian. Under such an approximation, the orbits are approximately temporally static. In such a case, the Hamiltonian is approximately time-independent. The Hamiltonian then would take the approximate form:

\[
H(\mathbf{q},\mathbf{p},t) \approx H_0(\mathbf{J})
\]

allowing for $H$ to be representable solely as a function of the constants of motion. For example, in tokamaks the equilbrium Hamiltonian (which can be identified as $K(\mathbf{J})$) is commonly taken to be of the form:

\[
H_0 = H_0(W, p_{\varphi}, \mu)
\]

where $W$ is the particle energy, $p_{\varphi}$ is the toroidal angular momentum, and $\mu$ is the particle magnetic moment. The corresponding conjugate quantities for which continuous symmetries approximately exist are time, azimuthal angle, and the gyroangle.

\section{Period of perturbed orbits}
\label{sec:period}

\begin{figure}[t!]
	\centering
	\includegraphics[width=0.8\textwidth]{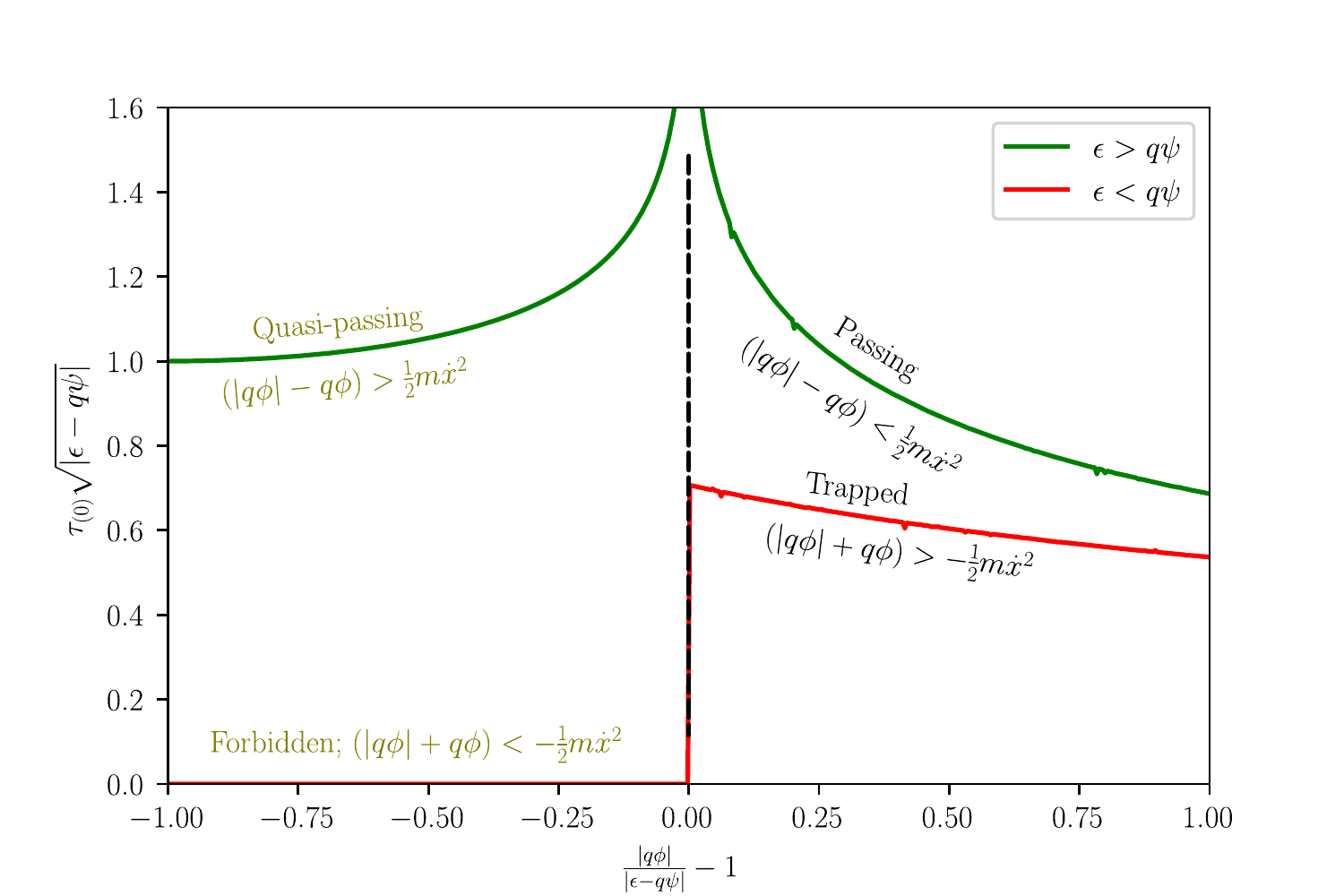}
	\caption{\textbf{Zeroth-order period $\tau_{(0)}$ for monochromatic waves as a function of \emph{wave amplitude}}. Particles with $\epsilon > q\psi$ are passing, while particles with $\epsilon < q\psi$ are trapped, directly continuing from the case where $q\psi = 0$ (conservative orbits). We categorize the `trapped' region of the plot to enable continuity from $q\psi = 0$, in agreement with the requirement that a particle would need an infinite amount of time to cross the barrier at $|q\phi| = |\epsilon - q\psi|$. The zeroth-order period $\tau_{(0)}$ decreases with wave amplitude $|q\phi|$ in the trapped region. In the passing region the period increases with wave amplitude, but in the quasi-passing region the period decreases with wave amplitude.}
	\label{fig:period_passing}
\end{figure}

As particles traverse phase space, they see an electric potential which varies in time and space. If one examines the Lorentz force for the $i^{\textrm{th}}$ particle:

\[
m \ddot{\mathbf{x}}^{(i)} = q [\mathbf{E} + \dot{\mathbf{x}}^{(i)} \times \mathbf{B}]_{\mathbf{x} = \mathbf{x}^{(i)}}
\]

where $m$ is the particle mass, $q$ is the particle charge, $\mathbf{E}(\mathbf{x},t) \in \mathbb{R}^3$ is the electric field, $\mathbf{B}(\mathbf{x},t) \in \mathbb{R}^3$ is the magnetic flux density, $\mathbf{x} \in \mathbb{R}^3$ is position, $t \in \mathbb{R}^{\geq 0}$ is time, and we employ Newton's notation such that $\dot{\mathbf{x}} := \partial \mathbf{x}/\partial t$. The position of the $i^{\textrm{th}}$ particle is denoted by $\mathbf{x}^{(i)}$. The work done by the magnetic field is zero, giving:

\begin{equation}
\label{eq:period_diff}
\dfrac{1}{2} m |\dot{\mathbf{x}}^{(i)}|^2 = -q \phi_{\mathbf{x} = \mathbf{x}^{(i)}} + U(t)
\end{equation}

where $\phi(\mathbf{x},t) \in \mathbb{R}$ is the electric potential, and $U(t)$ is the single particle energy. One can write this in the following form:

\begin{equation}
\label{eq:period_epsilon}
\epsilon(\mathbf{x}^{(i)},\dot{\mathbf{x}}^{(i)},t) = \dfrac{1}{2} m |\dot{\mathbf{x}}^{(i)}|^2 + q \phi_{\mathbf{x} = \mathbf{x}^{(i)}} + q \psi(\mathbf{x}^{(i)},\dot{\mathbf{x}}^{(i)},t)
\end{equation}

where $\epsilon$ and $\psi$ are energy-like functions, and $U(t) := \epsilon - q \psi$. For an orbit where the particle energy is conserved (herein referred to as a \emph{conservative} orbit) $U(t)$ must be constant. Via gauge freedom, we therefore define a conservative orbit as one where $q \psi = 0$. As such, $\epsilon$ is defined to be equal to the single particle energy when $U(t)$ is constant.

One can find the time taken to complete a closed orbit at a time $t$ by solving the differential equation given by \eqref{eq:period_diff}:

\[
\tau(t) = \oint_C \dfrac{\rmd |\mathbf{x}^{(i)}|(t)}{\sqrt{U(t) - q \phi(\mathbf{x}^{(i)}(t),t)}}
\]

where $C$ is a given orbit, and $\tau(t)$ is the orbit period. By performing Taylor expansion of the integrand about $U(t) = U(t_0)$:

\begin{equation}
\tau(t) = \oint_C \left\{\dfrac{\rmd |\mathbf{x}^{(i)}|(t)}{\sqrt{U(t_0) - q\phi(\mathbf{x}^{(i)}(t),t)}} \left[1 - \dfrac{1}{2}\left(\dfrac{\delta U(t)}{U(t_0) - q\phi(\mathbf{x}^{(i)}(t),t)} \right)  + \dfrac{3}{8}\left(\dfrac{\delta U(t)}{U(t_0) - q\phi(\mathbf{x}^{(i)}(t),t)} \right)^2 \dots \right] \right\}
\end{equation}

where $\delta U = U(t) - U(t_0)$. For $\tau$ to be finite on a given orbit:

\begin{align*}
U(t_0) \neq q \phi(\mathbf{x}^{(i)}(t),t) \, \forall \, \mathbf{x}^{(i)}(t) \in C &;&
\displaystyle \lim_{l \to \infty} \oint_C \dfrac{d |\mathbf{x}^{(i)}|(t)}{\sqrt{U(t_0) - q \phi(\mathbf{x}^{(i)}(t),t)}} \left(\dfrac{\delta U(t)}{\sqrt{U(t_0) - q \phi(\mathbf{x}^{(i)}(t),t)}}\right)^l = 0
\end{align*}

This yields an important limit for this perturbative orbit theory:

\begin{equation}
\delta U(t) < \sqrt{U(t_0) - q \phi(\mathbf{x}^{(i)}(t),t)}
\end{equation}

That is to say, for very large perturbations to the single particle energy, this perturbative theory breaks down; one cannot use this perturbative analysis any more as the Taylor expansion of the integrand is divergent. However, for all other cases, one can use this perturbative theory.

For further analysis, we drop the superscript $(i)$. One can therefore represent $\tau$ as $\tau_{(0)}(t) + \delta \tau(\delta U(t), t)$, where the zeroth-order contribution $\tau_{(0)}$ can be considered as the period of a near-conservative orbit; in a sense this is equivalent to the unperturbed orbit considered in \secref{sec:action}. This tells one the `instantaneous' period of the orbit that the particle is on; if $U(t)$ was to then remain fixed for all time $t > t_0$, the particle would continue on a conservative orbit with period $\tau_{(0)}$.

Suppose one examines a monochromatic potential wave given by $\phi(\mathbf{x},t) = |\phi|(t) \cos \left(kx - \int\limits_0^t \omega(t') \rmd t' + \theta\right)$, where $k$ is the wavenumber of the wave, $\omega(t)$ is the frequency of the wave, and $\theta$ is some arbitrary phase. Then, if one assumes that the bounce frequency of particles is much greater than $\omega$ (as is the case for a suitably large amplitude wave):

\begin{equation}
\label{eq:period_tau}
\tau_{(0)}(t) = \displaystyle\oint_C \dfrac{\rmd x}{\sqrt{U(t_0) - q |\phi|(t) \cos \left(kx - \int\limits_0^t \omega(t') \rmd t' + \theta\right)}} \approx \Re\left[\dfrac{4}{\sqrt{U(t_0) - q |\phi|(t)}} K\left(-\dfrac{2q |\phi|(t)}{U(t_0) - q |\phi|(t)} \right)  \right]
\end{equation}

\begin{figure}[t!]
	\centering
	\begin{subfigure}[t]{0.45\textwidth}
		\includegraphics[width=\textwidth]{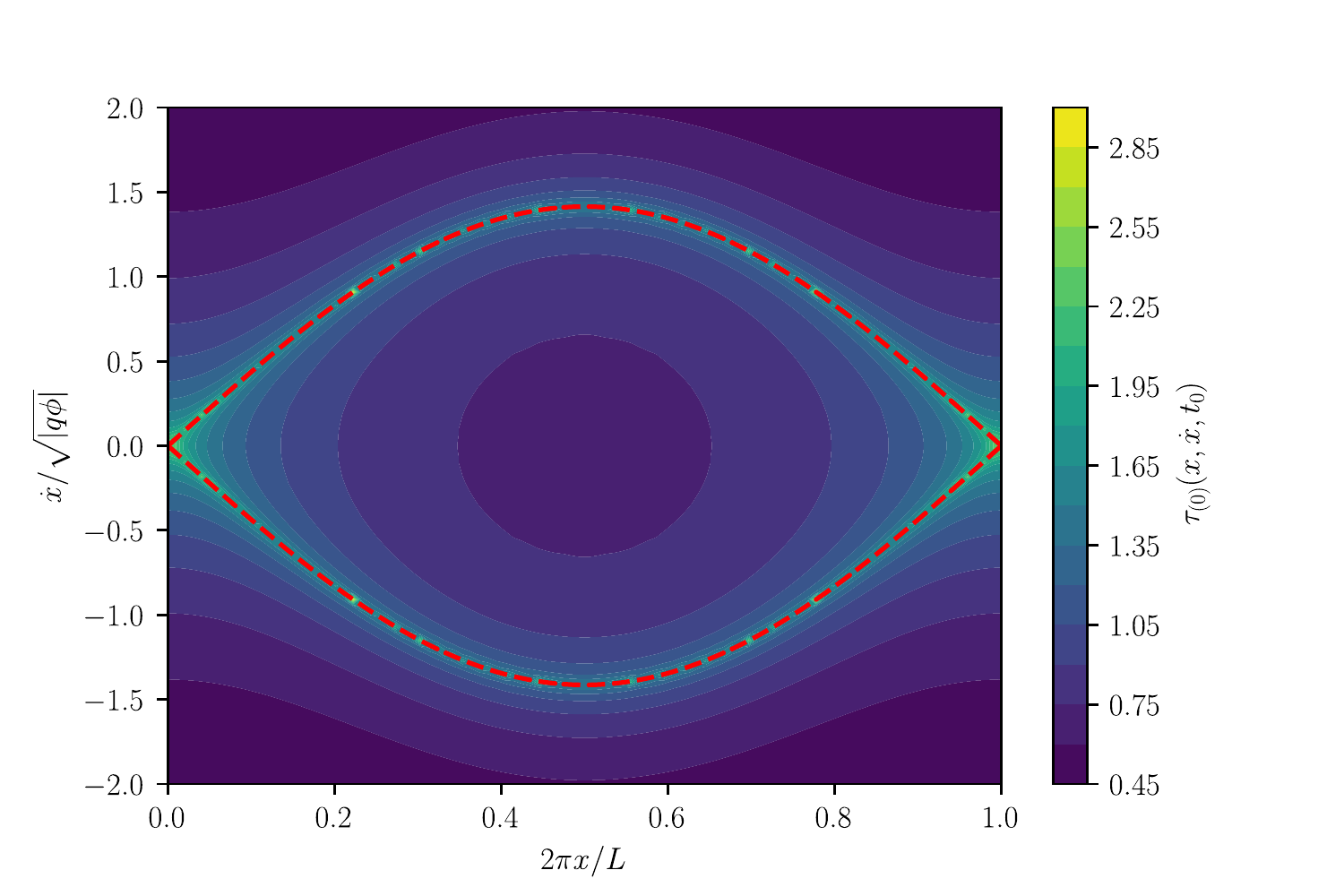}
		\caption{\textbf{Case $q\psi = 0$}. The conservative orbit contours are exactly that which is expected from BGK theory; each orbit corresponds to a different value of $\epsilon$. This is reflected by the fact that the zeroth-order contribution to the period, $\tau_{(0)}$, tends to infinity along the separatrix.}
		\label{fig:period_tau_psizero}
	\end{subfigure} ~~~ 
	\begin{subfigure}[t]{0.45\textwidth}
		\includegraphics[width=\textwidth]{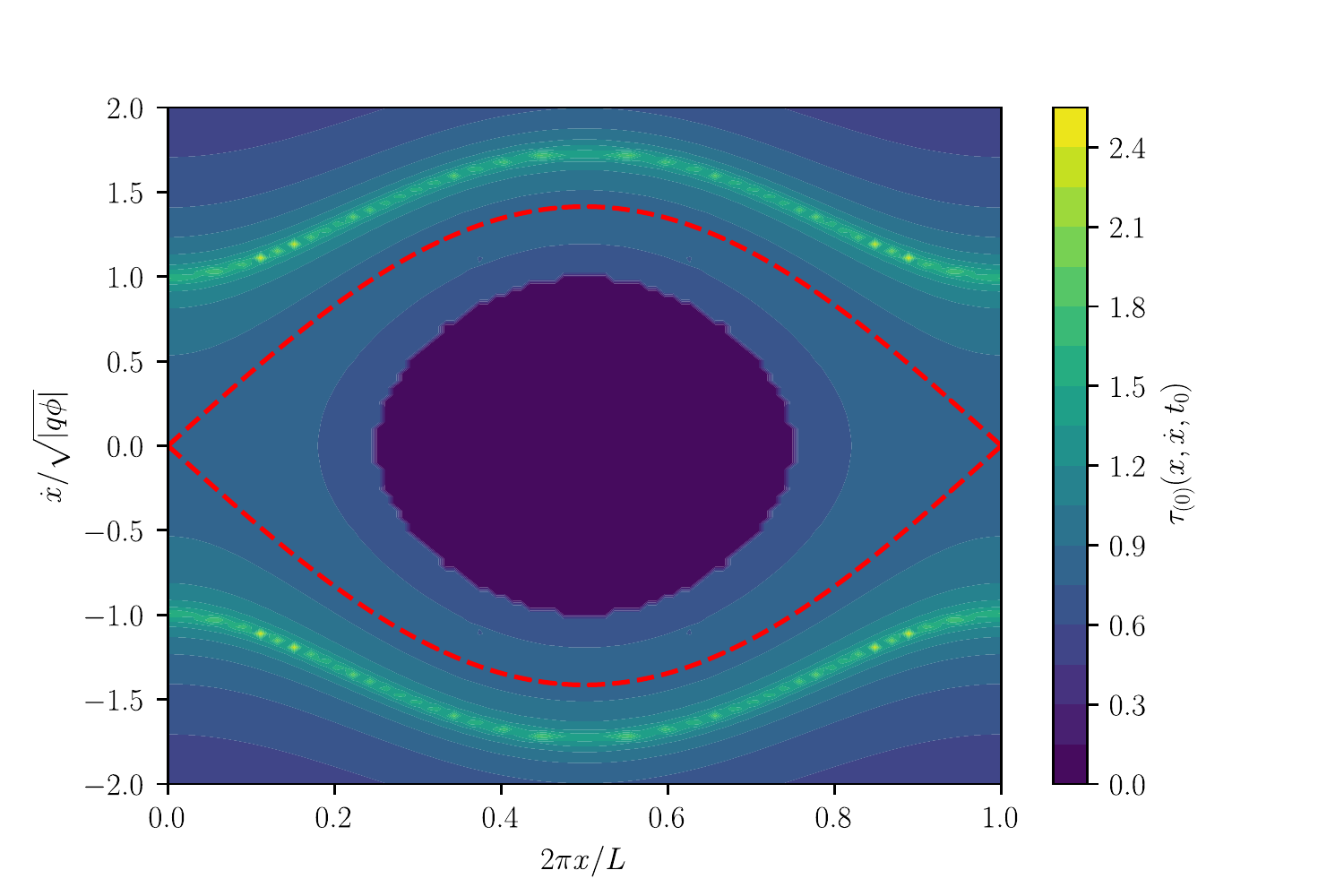}
		\caption{\textbf{Case $q\psi < 0$}. The contour where $\tau_{(0)} \to \infty$ occurs \emph{outside} of the separatrix. Accordingly, particles which are outside of the separatrix but within the infinite $\tau_{(0)}$ contour may exhibit similar behaviour to trapped particles (which we refer to as `quasi-passing' particles).}
		\label{fig:period_tau_psineg}
	\end{subfigure}
	\caption{\textbf{Period of particle-orbits in a BGK island.} Two plots illustrating conservative orbits in a single wavelength 1D potential $\phi(x,t) = |\phi| \cos (kx - \omega t)$ with phase velocity $u = \omega/k$. Filled contours correspond to different values of $\tau_{(0)}$, with the red line denoting the island separatrix. The length of the 1D box is given by $L = 2 \pi / k$.}
	\label{fig:period_tau}
\end{figure}

where $K(k)$ is the complete elliptic integral of the first kind. \cite{abramowitz1965handbook} Crucially, for all of the orbits, the period is different. Therefore, while the phase space structure is coherent (the entire structure moves with a single phase velocity), \emph{the single-particle orbits are not coherent}.

One can examine a single orbit such that $U(t=t_0) \equiv [\epsilon - q\psi]_{t = t_0}$ is fixed, and then investigate what may happen when one increases the wave amplitude. In \figref{fig:period_passing}, we plot the period of near-conservative orbits as a function of $q |\phi|$. The period increases monotonically with the wave amplitude provided that $\dot{x}^2 \gg q\phi$; these passing particles are almost coherent and approximately free stream through the phase space. One finds that the period is roughly constant for passing particles with $\dot{x}^2 > (q\phi-|q\phi|)$ but existing in a velocity width equal to the width of the phase space island, shown in the plot by the region where $\frac{|q\phi|}{|\epsilon - q\psi|} \to -\infty$; these particles are not trapped, but interact somewhat with the wave (quasi-passing). However, once the wave amplitude becomes large enough the period decreases monotonically with the wave amplitude; these particles fall into the potential as the wave grows.

The quasi-passing particles are of particular interest; while these particles are not trapped by the potential, they can still exchange energy with the wave. In \figref{fig:period_tau}, we show the orbit period $\tau_{(0)}$ as a function of particle position, particle velocity, and time. To achieve this, we calculate $\epsilon(x,\dot{x},t)$ using equation \eqref{eq:period_epsilon}. Then, using those values of $\epsilon$, we calculate $U(t_0)$. As such:

\[
\tau_{(0)}(t) = \tau_{(0)}[x(t),\dot{x}(t),t] = \tau_{(0)}[U(t_0),\phi(x,t)] = \tau_{(0)}[\epsilon(x,\dot{x},t),\psi,t]
\]

For the case where $q\psi = 0$ (see \figref{fig:period_tau_psizero}), contours where the orbit period tends to infinity align with the separatrix. However, when $q\psi < 0$ (see \figref{fig:period_tau_psineg}), contours where the orbit period tends to infinite lie outside of the separatrix. As such, particles which are outside of the separatrix may exhibit similar behaviour to trapped particles within the separatrix. These particles are quasi-passing, and while they are not trapped, they can be expected to exchange energy with the wave.

\section{Generating function formulation}
\label{sec:generating}
Here, we examine a collisionless plasma, using the full 6+1D Vlasov equation. We seek to find a solution using $\epsilon$ defined as the following:

\begin{equation}
\label{eq:generating_epsilon}
\epsilon = \sum\limits_l \left[q \phi^{[l]} + \dfrac{1}{2}m|\mathbf{v}-\mathbf{u}^{[l]}|^2 \right] + q \psi
\end{equation}

where $\phi$ is now the mean-field electric potential, and $\phi^{[l]}$ denotes a wavepacket of $\phi$ such that each constituent wave has phase velocity $u^{[l]}$. If one defines $f(x,v,t) = f[\epsilon]$, then, everywhere except on contours of constant $\epsilon$, $\epsilon$ satisfies the Vlasov equation:

\[
\hat{V} \epsilon = 0
\]

Where $\hat{V}$ is the `Vlasov operator':

\begin{equation}
\hat{V} := \pdif{}{t} + \mathbf{v} \cdot \nabla + \dfrac{q}{m} [\mathbf{E} + \mathbf{v} \times \mathbf{B}]\cdot \nabla_{\mathbf{v}}
\end{equation}

with $\mathbf{E}$ and $\mathbf{B}$ as the mean-field electric field and magnetic flux density respectively. We employ the following decomposition of the $\mathbf{E}$ and $\mathbf{B}$ field:\cite{woods2019analytical}

\begin{equation}
\label{eq:generating_E}
\left( \begin{array}{c} \mathbf{E} \\ \mathbf{B} \end{array} \right) (x,t) := \sum\limits_{l} \left( \begin{array}{c} \mathbf{E}^{[l]} \\ \mathbf{B}^{[l]} \end{array} \right) (x,t)
\end{equation}

where each wavepacket is given by:

\begin{equation}
\label{eq:generating_El}
\left( \begin{array}{c} \mathbf{E}^{[l]} \\ \mathbf{B}^{[l]} \end{array} \right) (x,t) := \dfrac{1}{2} \displaystyle\sum\limits_j \left\{\left( \begin{array}{c} \mathbf{E}_j^{[l]} \\ \mathbf{B}_j^{[l]} \end{array} \right)  \exp \left[\rmi \mathbf{k}_j \cdot \mathbf{x} + \int\limits_{0}^{t} p_{j}^{[l]} (\tau) \rmd \tau\right] + \textrm{c.c.} \right\}
\end{equation}

and where each $p_j^{[l]} \in \mathbb{C}$ is given by:

\[
p_j^{[l]} := \gamma_j^{[l]} - \rmi \omega_j^{[l]}
\]

If one substitutes $\epsilon$ into the Vlasov equation:

\begin{equation}
\label{eq:generating_vlasov}
\begin{array}{r l}
q \hat{V} \psi = \overbrace{\sum\limits_l \left[m(\mathbf{v}-\mathbf{u}^{[l]}) \cdot \left\{\pdif{\mathbf{u}^{[l]}}{t} + \mathbf{v} \cdot \nabla \mathbf{u}^{[l]} \right\}\right]}^{\textrm{frequency sweep}} - \underbrace{q\sum\limits_l \left[ \pdif{\phi^{[l]}}{t} + \mathbf{u}^{[l]} \cdot \nabla \phi^{[l]} \right]}_{\textrm{drive}} + \underbrace{\sum\limits_l \left[ \mathbf{u}^{[l]} \cdot (\mathbf{v} \times q \mathbf{B}) \right]}_{\textrm{gyration}}
\end{array}
\end{equation}

In this sense, the generating function $\psi$ provides a correction to $\epsilon$. As contours of constant $f$ (and therefore constant $\epsilon$) denote particle orbits, systems with finite $\{\gamma_j^{[l]}\}$, finite $\{\partial_t \omega_j^{[l]}\}$, finite curvature of $\{\mathbf{k}_j\}$, or collisions require particles to take nonconservative orbits in phase space. 

The `frequency sweep' term and `drive' term correspond to advection of $\epsilon_0$. Using the form of the electric potential given by \label{eq:generating_El}, the electric potential is a solution of the coupled inhomogeneous advection equation:

\begin{equation}
\label{eq:generating_inhom}
\begin{array}{r l}
\displaystyle\sum\limits_l \left[\pdif{\phi^{[l]}}{t} + \mathbf{u}^{[l]} \cdot \nabla \phi^{[l]} \right] &= \displaystyle\sum\limits_{j,l} \Bigg\{ \left\{ \gamma_j^{[l]} - \mathbf{u}^{[l]} \cdot [\nabla(\mathbf{k}_j \cdot \mathbf{x}) - \mathbf{k}_j)] \right\} |\phi_j^{[l]}| \\
& \hspace{30pt} \displaystyle \cdot \exp \left[\int_0^t \gamma_j^{[l]} \rmd \tau\right] \cos \left[\mathbf{k}_j  \cdot \mathbf{x} - \int_0^{t} \omega_j^{[l]} \rmd \tau + \theta_j^{[l]}\right] \Bigg\}
\end{array}
\end{equation}

where $\theta_j^{[l]}$ is a phase angle. The term proportional to $ \mathbf{u}^{[l]} \cdot [\nabla(\mathbf{k}_j \cdot \mathbf{x}) - \mathbf{k}_j]$ is a correction to the drive that is given by the curvature of the wavevector. If the wave is propagating  in a rectilinear fashion, this term is identically zero.

In \secref{subsec:generating_transformation}, we give Galilean transformation matrices which allow one to represent the problem in non-inertial, comoving frames. Then, in \secref{subsec:generating_frequency} and \secref{subsec:generating_growth}, we give functions which generate the `frequency sweep' and `growth' terms in \eqref{eq:generating_vlasov}. Later, in \secref{subsec:generating_nowave}, we give approximate solutions for the electrostatic Vlasov equation, under the limit that no wave-wave coupling occurs.

\subsection{Transformation matrices}
\label{subsec:generating_transformation}
Here, we utilise the following Galilean transformation matrix \cite{woods2019analytical}:

\begin{equation}
\label{eq:generating_transformation_chi}
\chi_j^{[l]} := \mathbf{k}_j \cdot \mathbf{x} - \int\limits_0^{t} \omega_j^{[l]}(\tau) \rmd \tau
\end{equation}

such that each element of the matrix corresponds to different wavepackets (different $\mathbf{u}^{[l]}$) and different wavevectors. This is motivated by BGK theory, where BGK modes are derived in the comoving frame for the wave. Logically, one can also define a corresponding velocity transformation matrix $\nu$:

\begin{equation}
\label{eq:generating_transformation_nu}
\nu_{j}^{[l]} := \mathbf{k}_j \cdot \mathbf{v} - \omega_{j}^{[l]}
\end{equation}

Components of the matrix are zero at the point of Landau resonance, where $\omega_j^{[l]} = \mathbf{v} \cdot \mathbf{k}_j$. The derivatives transform covariantly under $(\mathbf{x},\mathbf{v},t) \to (\chi, \nu, t)$, and therefore one must careful perform the co-ordinate transformation as the direction of wave propagation can be curvilinear. Therefore:

\begin{subequations}
\label{eq:generating_transformation_derivatives}
\begin{align}
\nabla &= \displaystyle\sum\limits_{j,l} \left\{ \mathbf{k}_j - [\nabla(\mathbf{k}_j \cdot \mathbf{x}) - \mathbf{k}_j] \right\}\left.\pdif{}{\chi_{j}^{[l]}}\right|_{\nu,t} \\
\nabla_{\mathbf{v}} &= \displaystyle\sum\limits_{j,l} \mathbf{k}_j \left.\pdif{}{\nu_{j}^{[l]}}\right|_{\nu,t} \\
\left.\pdif{}{t}\right|_{\mathbf{x},\mathbf{v}} &= -\displaystyle\sum\limits_{j,l} \left\{\omega_j^{[l]} \left.\pdif{}{\chi_{j}^{[l]}}\right|_{\nu,t} + \ddif{\omega_j^{[l]}}{t} \left.\pdif{}{\nu_{j}^{[l]}}\right|_{\chi,t} \right\} +  \left.\pdif{}{t}\right|_{\chi,\nu}
\end{align}
\end{subequations}

Under this transformation, one can split $\psi$ into three parts: $\psi_{\partial}$ encapsulating all of the curvilinear terms (given by $\sim [\nabla(\mathbf{k}_j \cdot \mathbf{x}) - \mathbf{k}_j]$), $\psi_{\mathbf{B}}$ encapsulating all of the magnetic field terms, and $\psi_0$ containing only the rectilinear terms and the electric field parts. Then:

\begin{equation}
\label{eq:generating_transformation_psi0}
\begin{array}{l}
q\left[\left.\pdif{}{t}\right|_{\chi,\nu} + \displaystyle\sum\limits_{j,l} \bigg\{\nu_{j}^{[l]} \pdif{}{\chi_{j}^{[l]}} - \dfrac{q}{m} \left[\dfrac{m}{q} \ddif{\omega_j^{[l]}}{t} + |\mathbf{k}_j|^2 \sum\limits_{l'} \pdif{\phi}{\chi_{j}^{[l']}} \right] \pdif{}{\nu_{j}^{[l]}}\bigg\} \right] \psi_0 \\
\hspace{40pt} = \displaystyle\sum\limits_{j,l} \left[m \nu_j^{[l]} \ddif{\omega_j^{[l]}}{t} \right] - q\sum\limits_l \left.\pdif{\phi^{[l]}}{t}\right|_{\chi,\nu}
\end{array}
\end{equation}

\subsection{Frequency sweep generating function}
\label{subsec:generating_frequency}
Here, we aim to solve the equation:

\begin{equation}
\label{eq:generating_frequency_eq}
q\displaystyle\sum\limits_{j,l} \bigg\{\nu_{j}^{[l]} \pdif{}{\chi_{j}^{[l]}} - \dfrac{q}{m} |\mathbf{k}_j|^2 \sum\limits_{l'} \pdif{\phi}{\chi_{j}^{[l']}}\pdif{}{\nu_{j}^{[l]}}\bigg\} \psi_{\textrm{sw.}} = \displaystyle\sum\limits_{j,l} \left[m \nu_j^{[l]} \ddif{\omega_j^{[l]}}{t} \right] 
\end{equation}

This allows one to generate the `frequency sweep' term in equation \eqref{eq:generating_vlasov}. As the right hand side is independent of $\chi_j^{[l]}$, it is fairly straightforward to show that this has a solution given by:

\begin{equation}
\label{eq:generating_frequency_psisw}
\psi_{\textrm{sw.}} = \sum\limits_{j,l} \left[\dfrac{m}{q} \chi_j^{[l]} \ddif{\omega_j^{[l]}}{t} \right]
\end{equation}

This solution is not periodic; we will address this later in \secref{subsec:generating_nowave}.

\subsection{Growth rate generating function}
\label{subsec:generating_growth}
Here, we aim to solve the equation:

\begin{equation}
\label{eq:generating_growth_eq}
q\displaystyle\sum\limits_{j,l} \bigg\{\nu_{j}^{[l]} \pdif{}{\chi_{j}^{[l]}} - \dfrac{q}{m} |\mathbf{k}_j|^2 \sum\limits_{l'} \pdif{\phi}{\chi_{j}^{[l']}}\pdif{}{\nu_{j}^{[l]}}\bigg\} \psi_{\textrm{sw.}} = - q\sum\limits_l \left.\pdif{\phi^{[l]}}{t}\right|_{\chi,\nu}
\end{equation}

This allows one to generate the `drive' term in equation \eqref{eq:generating_vlasov}.

\subsubsection{Single wavepacket of constant frequency}
\label{subsubsec:generating_growth_single}
One can first attempt to solve the following equation:

\begin{equation}
\label{eq:generating_growth_single_eq}
\left[\mathbf{v} \cdot \nabla - \dfrac{q}{m} (\nabla \phi) \cdot \nabla_{\mathbf{v}} \right] \psi = - \gamma \phi
\end{equation}

where $\phi = \phi(\mathbf{x})$. To do so, one requires an extension of the Leibniz integral rule to vector differential operators (see \appref{app:leibniz}). Using the vectorial form, one finds that the following integral is useful:

\[
I_n := \dfrac{1}{v} \int\limits_{\mathbf{x}_{(0)}}^{\mathbf{x}} \phi(\mathbf{x}')^n \mathbf{v} \cdot \rmd \mathbf{x}'
\]

where $\mathbf{x}_{(0)}$ is selected via gauge freedom. $I_n$ yields the following gradient:

\[
\nabla I_n = \dfrac{\mathbf{v}}{|\mathbf{v}|} \phi(\mathbf{x})^n
\]

In \appref{app:infinite}, we therefore show that the following solution for $\psi$ is permitted:

\begin{equation}
\label{eq:generating_growth_single_psi}
\psi = - \gamma \sqrt{\dfrac{m}{2}} \int\limits_{\mathbf{x}_{(0)}}^{\mathbf{x}}\dfrac{q\phi(\mathbf{x}')}{\sqrt{\epsilon_{(0)}(\mathbf{x},\mathbf{v}) - q \phi(\mathbf{x}')}} \dfrac{\mathbf{v}}{|\mathbf{v}|} \cdot \rmd \mathbf{x}'
\end{equation}

where $\epsilon_0 = q \phi + \frac{1}{2}mv^2$. It is important to note that $\psi$ is real and without singularities in the region:

\[
q \phi(x)  + \dfrac{1}{2}mv^2 > q |\phi|
\]

where $|\phi|$ is the amplitude of the wave. For values where $q \phi(\mathbf{x}) + \frac{1}{2}m|\mathbf{v}|^2 = q |\phi|$, $\psi$ is singular. Furthermore, for $q \phi(\mathbf{x}) + \frac{1}{2}m|\mathbf{v}|^2 < q |\phi|$, $\psi$ is imaginary. Therefore, to enable real values of $\psi$, one must find another solution for the regions where $q \phi(\mathbf{x}) + \frac{1}{2}m|\mathbf{v}|^2 \leq q |\phi|$.

In addition, the solution for $\psi$ given above is not periodic. We will address this later in \secref{subsec:generating_nowave}.

\subsubsection{Multiple non-interacting wavepackets}
\label{subsubsec:generating_growth_multiple}
The full equation given by [CITE] is very difficult to solve. Instead, one can examine a simpler scenario by enforcing no wave-wave coupling here and neglecting the time derivative:

\[
\left[\sum\limits_{j,l} \bigg\{\nu_{j}^{[l]} \pdif{}{\chi_{j}^{[l]}} - \dfrac{q}{m} |\mathbf{k}_j|^2 \pdif{\tilde{\phi}_j^{[l]}}{\chi_{j}^{[l]}} \pdif{}{\nu_{j}^{[l]}}\bigg\} \right] \psi_{\gamma} \approx - \sum\limits_{j,l} \gamma_j^{[l]} \tilde{\phi}_j^{[l]}
\]

where $\tilde{\phi}_j^{[l]} = \frac{1}{2} \exp \left[\rmi \mathbf{k}_j \cdot \mathbf{x} - \int\limits_0^{t} p_j^{[l]} \rmd \tau \right] + \textrm{c.c.}$, and all terms in the inner sum with $l \neq l'$ have been discarded. It is trivial to show that by extending \eqref{eq:generating_growth_single_eq} to an arbitrarily sized dimensional space and solving for each $\psi_j$ independently, a solution is given by:

\begin{equation}
\label{eq:generating_growth_multiple_psig}
\boxed{
\psi_{\gamma} = \psi_{\textrm{ww}} - \sum\limits_{j,l} \gamma_j^{[l]} \dfrac{\nu_{j}^{[l]}}{\sqrt{\sum\limits_{j,l} (\nu_{j}^{[l]})^2}} \sqrt{\dfrac{m}{2}} \int\limits_{\chi_{(0),j}^{[l]}}^{\chi_j^{[l]}}\dfrac{q\phi(\chi_j^{[l]}\,',\dots)}{\sqrt{\epsilon_{(0)}(\chi,\nu) - q \phi(\chi_j^{[l]}\,',\dots)}} \rmd \chi_j^{[l]}\,'
}
\end{equation}

where $\{\chi_{(0),j}^{[l]}\}$ are selected via gauge freedom, such that one simply sums over all wavevectors and frequencies independently. $\psi_{\textrm{ww}}$ encapsulates extra terms that would appear which are highly nonlinear, and feature wave-wave coupling. These interactions cannot be neglected for wave-wave interactions.

\subsection{No wave-wave coupling}
\label{subsec:generating_nowave}
For collisionless systems, one can represent the generating function $\psi$ as the sum of three terms:

\begin{equation}
\label{eq:generating_nowave_psi}
\psi = \psi_{\gamma} + \psi_{\textrm{sw.}} + \delta \psi
\end{equation}

where $\delta \psi$ is a highly nonlinear term that is given by all other contributions:

\begin{equation}
\label{eq:generating_nowave_dpsi}
\delta \psi := \psi_{\partial} + \psi_{\mathbf{B}} + \psi_{\textrm{ww}} + [\psi_0(\chi,\nu,t) - \psi_0(\chi,\nu,t=0)]
\end{equation}

By operating on $\psi$, in a system with rectilinearly propagating waves one finds by inspection, under constant growth rates:

\[
\hat{V}\delta \psi = \sum\limits_{j,l} \left[\dfrac{m}{q} \chi_j^{[l]} \ddif{^2 \omega_j^{[l]}}{t^2} \right] - \sum\limits_{j,l} \gamma_j^{[l]} \sqrt{\dfrac{m}{2}} \int\limits_{\chi_{(0),j}^{[l]}}^{\chi_j^{[l]}}\dfrac{q\phi(\chi_j^{[l]}\,',\dots)}{2(\epsilon_{(0)}(\chi,\nu) - q \phi(\chi_j^{[l]}\,',\dots))^{3/2}}  \hat{V}\left( \epsilon_{(0)} \dfrac{\nu_{j}^{[l]}}{\sqrt{\sum\limits_{j,l} (\nu_{j}^{[l]})^2}} \right) \rmd \chi_j^{[l]}\,'
\]

As such, if the frequency sweep rate is static, there is a small $\mathbf{B}$-field and $\epsilon_0$ is approximately static, then $\delta \psi$ can be considered to be mostly a small perturbation to the system; therefore:

\begin{equation}
\exists \, \mathbf{E}, \mathbf{B} \, : \, \hat{V} \delta \psi \ll \hat{V} (\psi - \delta \psi)
\end{equation}

However, there is always a finite amount of wave-wave coupling in the system, as growth rate $\to$ modified orbit shape $\to$ shear $\to$ wave-wave coupling. Intrinsically, by discarding this term one assumes low shear of phase space islands, and negligible coupling between islands. One can assume that islands which are far apart do not interact with each other when considering gap toroidal Alfv\'{e}n eigenmodes in tokamaks,\cite{heidbrink2008basic} however this approximation does not allow one to accurately examine hole and clump formation. For one to do so, one must retain all of the wave-wave coupling encapsulated by $\delta \psi$.

With regards to $(\psi - \delta \psi)$, in the limit of $\delta \psi \to 0$, one desires periodicity:

\[
\lim_{\delta \psi \to 0} (\psi - \delta \psi) \, : \, \psi(\mathbf{x},\mathbf{v},t) \equiv \psi(\mathbf{x}+ L_i \hat{\mathbf{x}}_i,\mathbf{v},t) \, \forall \, i
\]

where $L_i$ is the box length in the $x_i$ direction. To enforce this, one can demand that all terms proportional to $\chi_j^{[l]}$ must vanish. This requires:

\begin{equation}
\label{eq:generating_nowave_cond}
\boxed{
\sum\limits_{j,l} \dfrac{m}{q} \ddif{\omega_j^{[l]}}{t} \approx \sum\limits_{j,l} \gamma_j^{[l]} \dfrac{\nu_{j}^{[l]}}{\sqrt{\sum\limits_{j,l} (\nu_{j}^{[l]})^2}} \sqrt{\dfrac{m}{2}} \int\limits_{\chi_{(0),j}^{[l]}}^{\chi_{(0),j}^{[l]} + 2\pi}\dfrac{q\phi(\chi_j^{[l]}\,',\dots)}{\sqrt{\epsilon_0(\chi,\nu) - q \phi(\chi_j^{[l]}\,',\dots)}} \rmd \chi_j^{[l]}\,' \approx 0
}
\end{equation}

This can be safely enforced for the case where the sweeping rate is small, and where the wave amplitude is small. In such a scenario, one can approximate:

\[
\psi \approx \psi_{\gamma} + \psi_{\textrm{sw.}}
\]

\begin{figure}[h!]
	\centering
	\includegraphics[width=0.6\textwidth]{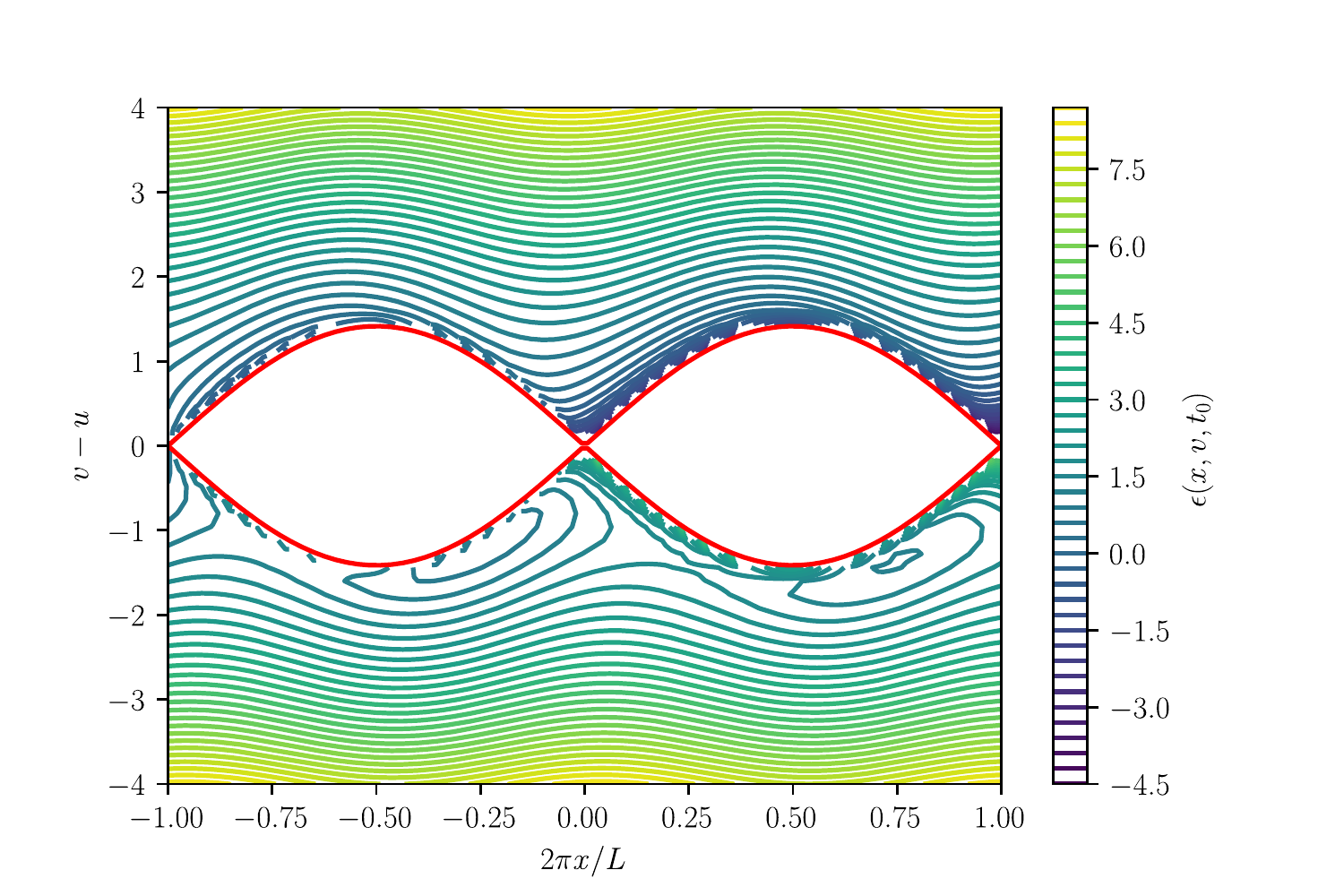}
	\caption{\textbf{Existence of quasi-passing particles}. Contours of constant $\epsilon(x,v,t)$ for a single 1D electrostatic wave with time-invariant frequency and growing amplitude undergoing a form of resonant interaction in a 1D kinetic system (derived in \secref{subsec:generating_nowave}). Particles outside of the separatrix appear to exist in closed orbits; it is possible that these are quasi-passing particles as proposed in \secref{subsec:generating_nowave}.}
	\label{fig:generating_nowave_psi}
\end{figure}

In \figref{fig:generating_nowave_psi}, we show an example case with a single electric potential wave with time-invariant frequency and growing amplitude. Close to the separatrix, some particles appear to be existing in closed orbits. It is worth noting that the condition as posed in \eqref{eq:generating_nowave_cond} is violated here. As a result, there is a small amount of aperiodicity in the plots. In reality, this aperiodicity must be cancelled out by $\delta \psi$, as $\psi$ must be periodic.

\section{Conclusions}
In this paper, we examined analytical extensions to BGK theory that allow for finite growth rate and frequency sweeping without considering a `continuum of BGK modes'. We examined how the period of orbits is modified under perturbations to the electric potential, showing that a population of particles act as `quasi-passing', where they are not bound by the electric potential, but interact with the wave. Using this solution, in \secref{subsec:generating_nowave} we gave approximate solutions for $\epsilon$. In \figref{fig:generating_nowave_psi}, we showed that a subpopulation of particles outside of the separatrix appear to embark on orbits which straddle the separatrix. Heuristically, we posit that these may be `quasi-passing' particles as proposed in \secref{sec:period}.

By considering $\epsilon$ along particle orbits as the sum of a quantity representing the single particle energy and a generating function $\psi$, we showed that an analytical closed form solution exists which allows for finite growth rate. 

A next logical step would be to utilise this theory to recover analytic predictions such as those given by Berk-Breizman models for the form of phase-space structures during events such as frequency bifurcation, mode growth, and frequency sweeping. However, this would require one to re-examine the solution given in \secref{subsec:generating_nowave}, as one would need to allow the growth rates to vary in time. In addition in Berk-Breizman models, the kinetic equation that describes the evolution of the distribution function in these models features other sources and sinks which have been omitted here.\cite{berk1997spontaneous} Furthermore, the approximate solutions given here violate periodic boundary conditions, and therefore a further step would be to seek solutions which do not include this aperiodicity. These solutions would include extra terms omitted here.

\section{Acknowledgements}
The author was funded by the EPSRC Centre for Doctoral Training in Science and Technology of Fusion Energy grant EP/L01663X. This work has been carried out within the framework of the EUROfusion Consortium and has received funding from the Euratom research and training programme 2014-2018 and 2019-2020 under grant agreement No 633053. The views and opinions expressed herein do not necessarily reflect those of the European Commission. 

The author would like to thank V. N. Duarte and R. G. L. Vann for useful discussions.

\appendix
\section{Infinite sum solution for growth rate generating function}
\label{app:infinite}
Here we shall solve equation \eqref{eq:generating_growth_single_eq}.

\begin{lemma}
The following equation:

\[
\left[\mathbf{v} \cdot \nabla - \dfrac{q}{m} (\nabla \phi) \cdot \nabla_{\mathbf{v}} \right] \psi = - \gamma \phi
\]

permits a solution:

\[
\psi = - \dfrac{\gamma}{|\mathbf{v}|} \sum\limits_{n=0}^{\infty} \left( \dfrac{q}{\frac{1}{2}mv^2} \right)^n \dfrac{(-1)^n (2n!)}{2^{2n} n!} C_n
\]

where $C_n$ are integral polynomials:

\[
C_n \equiv \sum\limits_{l=0}^{n} \dfrac{(-1)^l}{l(n-l)!} \phi^{n-l} I_{l+1}
\]

with $I_n(x)$ defined as:

\begin{equation}
I_n(x) := \int_{\mathbf{x}_{(0)}}^{\mathbf{x}} \phi^n(x') \dfrac{\mathbf{v}}{|\mathbf{v}|} \cdot \rmd \mathbf{x}'
\end{equation}
\end{lemma}

\begin{proof}
Seek a term which when operated on with $\mathbf{v} \cdot \nabla$ returns the right hand side, $- \gamma \phi$. Then, operate on this with the left hand side operator and see what extra term is generated. By using $\gamma I_1/|\mathbf{v}|$, one finds:

\[
\left[\mathbf{v} \cdot \nabla - \dfrac{q}{m} (\nabla \phi) \cdot \nabla_{\mathbf{v}} \right] \left( \psi + \gamma \dfrac{I_1}{|\mathbf{v}|} \right) = \gamma \dfrac{q}{m} \pdif{\phi}{x} \dfrac{I_1}{|\mathbf{v}|^2}
\]

where the term on the right hand side is generated by the Lorentz force term ($\nabla_{\mathbf{v}}$). Iteratively, one can find the series by repeating this technique:

\[
\left[\mathbf{v} \cdot \nabla - \dfrac{q}{m} (\nabla \phi) \cdot \nabla_{\mathbf{v}} \right]  \left( \psi + \gamma \left[ \dfrac{I_1}{|\mathbf{v}|} -  \dfrac{q}{m} \dfrac{\phi I_1 - I_2}{|\mathbf{v}|^3} \right] \right) = \gamma \dfrac{q^2}{m^2} \pdif{\phi}{x} \dfrac{3 (\phi I_1 - I_2)}{|\mathbf{v}|^4}
\]

At this point the following identity becomes useful:

\[
(a+1) \phi^{a} I_b \nabla \phi = \nabla [\phi^{a+1} I_b - I_{a+b+1}]
\]

One finds that on each iteration one produces a term which is $q/mv^2$ times an integral polynomial. Therefore, after an infinite number of iterations, one finds:

\begin{equation}
\label{app:nonconservative_growth_infinite1D_iteration}
\left[\mathbf{v} \cdot \nabla - \dfrac{q}{m} (\nabla \phi) \cdot \nabla_{\mathbf{v}} \right] \psi' \equiv \lim_{n \to \infty} \gamma \left( \dfrac{q}{\frac{1}{2}m|\mathbf{v}|^2} \right)^n \dfrac{(-1)^n (2n!)}{2^{2n} n!} C_n
\end{equation}

where $\psi'$ is given by:

\[
\psi' = \psi + \dfrac{\gamma}{|\mathbf{v}|} \sum\limits_{n=0}^{\infty} \left( \dfrac{q}{\frac{1}{2}m|\mathbf{v}|^2} \right)^n \dfrac{(-1)^n (2n!)}{2^{2n} n!} C_n
\]

One finds that via l'Hopital's rule, the right hand side of \eqref{app:nonconservative_growth_infinite1D_iteration} vanishes everywhere except at $v = 0$. 
Then, via gauge freedom one can choose $\psi'$ to be zero. Therefore:

\[
\psi = - \dfrac{\gamma}{|\mathbf{v}|} \sum\limits_{n=0}^{\infty} \left( \dfrac{q}{\frac{1}{2}m|\mathbf{v}|^2} \right)^n \dfrac{(-1)^n (2n!)}{2^{2n} n!} C_n
\]
\end{proof}

Next, we shall show that this has a closed form solution as a single integral function:

\begin{lemma}
\label{app:nonconservative_growth_infinite1D_closedform}
The following function:

\[
\psi = - \dfrac{\gamma}{|\mathbf{v}|} \sum\limits_{n=0}^{\infty} \left( \dfrac{q}{\frac{1}{2}m|\mathbf{v}|^2} \right)^n \dfrac{(-1)^n (2n!)}{2^{2n} n!} C_n
\]

is representable in the form:

\[
\psi = - \gamma \sqrt{\dfrac{m}{2}} \int\limits_{\mathbf{x}_{(0)}}^{\mathbf{x}}\dfrac{q\phi(\mathbf{x}')}{\sqrt{\epsilon_{(0)}(\mathbf{x},\mathbf{v}) - q \phi(\mathbf{x}')}} \dfrac{\mathbf{v}}{|\mathbf{v}|} \cdot \rmd \mathbf{x}'
\]

where $\epsilon_{(0)}$ is defined as:

\[
\epsilon_{(0)}(\mathbf{x},\mathbf{v}) = q\phi(\mathbf{x}) + \dfrac{1}{2}m|\mathbf{v}|^2
\]
\end{lemma}

\begin{proof}
First, to switch the order of the summation:

\[
\begin{array}{r c l c r c l}
n = 0 &;& l \in [0,0] &\to& l = 0 &:& n \in [0,\infty] \\
n = 1 &;& l \in [0,1] &\to& l = 1 &:& n \in [1,\infty] \\
n = 2 &;& l \in [0,2] &\to& l = 2 &:& n \in [2,\infty] \\
\vdots && \vdots && \vdots && \vdots
\end{array}
\]

Therefore one finds that switching the order of summation requires:

\[
\sum\limits_{n=0}^{\infty} \sum\limits_{l=0}^{n} \to \sum\limits_{l=0}^{\infty} \sum\limits_{n=l}^{\infty} 
\]

With this information, one can evaluate the $n$ sum first:

\[
\psi = - \dfrac{\gamma}{|\mathbf{v}|} \sum\limits_{l=0}^{\infty} \sum\limits_{n=l}^{\infty} \left( \dfrac{q}{\frac{1}{2}m|\mathbf{v}|^2} \right)^n \dfrac{(-1)^n (2n!)}{2^{2n} n!} \dfrac{(-1)^l}{l(n-l)!} \phi^{n-l} I_{l+1}
\]

This allows one to find a function of $\phi$ which is \emph{independent} of integration. Noting that:

\begin{align*}
\dfrac{\phi^{n-l}}{(n-l)!} &= \dfrac{(n-l+1)(n-l+2) \dots (n)}{n!} \phi^{n-l} \\
&= \dfrac{\phi^{n-l}}{n!} \prod\limits_{r=1}^{l} (n-l+r) \\
&= \dfrac{1}{n!} \left( \dfrac{\rmd}{\rmd \phi} \right)^l \phi^n
\end{align*}

it is possible to show that:

\begin{align*}
\psi &= - \dfrac{\gamma}{|\mathbf{v}|} \sum\limits_{l=0}^{\infty} \dfrac{(-1)^l}{l!} I_{l+1} \left( \dfrac{\rmd}{\rmd \phi} \right)^l \sum\limits_{n=l}^{\infty} R^n \dfrac{(-1)^n (2n!)}{2^{2n} (n!)^2} \\
&=- \dfrac{\gamma}{|\mathbf{v}|} \sum\limits_{l=0}^{\infty} \dfrac{(-1)^l}{l!} I_{l+1} \left( \dfrac{q}{\frac{1}{2}m|\mathbf{v}|^2} \dfrac{\rmd}{\rmd R} \right)^l \sum\limits_{n=0}^{\infty} R^n \dfrac{(-1)^n (2n!)}{2^{2n} (n!)^2}
\end{align*}

where we have used the shorthand:

\[
R := \dfrac{q \phi}{\frac{1}{2}m|\mathbf{v}|^2}
\]

This quantity has physical significance; it is the ratio of the potential energy and kinetic energy of the particle. The sum over $n$ is nothing more than a fractional binomial expansion:

\[
\sum\limits_{n=0}^{\infty} R^n \dfrac{(-1)^n (2n!)}{2^{2n} (n!)^2} = (1 + R)^{-0.5}
\]

and therefore one now finds:

\[
\psi = - \dfrac{\gamma}{|\mathbf{v}|} \sum\limits_{l=0}^{\infty} \dfrac{(-1)^l}{l!} I_{l+1} \left( \dfrac{q}{\frac{1}{2}m|\mathbf{v}|^2} \dfrac{\rmd}{\rmd R} \right)^l (1 + R)^{0.5}
\]

By induction, one finds:

\[
\begin{array}{r l}
\left( \dfrac{\rmd}{\rmd R} \right)^l (1 + R)^{0.5} &= \left( \dfrac{\rmd}{\rmd R} \right)^{l-1} (1 + R)^{-1.5} \left( \dfrac{-1}{2} \right) \\
&= \left( \dfrac{\rmd}{\rmd R} \right)^{l-2} (1 + R)^{-2.5} \left( \dfrac{-1}{2} \right) \left( \dfrac{-3}{2} \right)\\
&= (1 + R)^{-l - 0.5} \dfrac{(-1)^l (2l)!}{2^{2l} l!}
\end{array}
\]

and therefore:

\[
\psi = - \dfrac{\gamma}{|\mathbf{v}|} \sum\limits_{l=0}^{\infty}  I_{l+1} \left( \dfrac{q}{\frac{1}{2}m|\mathbf{v}|^2} \right)^l (1 + R)^{-l - 0.5} \dfrac{(2l)!}{2^{2l} (l!)^2}
\]

It is possible to combine terms by noting the following:

\[
\dfrac{(1 + R)^{-l-0.5}}{(\frac{1}{2}m|\mathbf{v}|^2)^l} = \sqrt{\dfrac{m|\mathbf{v}|^2}{2 \epsilon_{(0)}}} (\epsilon_{(0)})^{-l}
\]

Therefore, $\psi$ now takes the form:

\[
\begin{array}{r l}
\psi &= - \gamma \sqrt{ \dfrac{m}{2 \epsilon_{(0)}} } \displaystyle\sum\limits_{l=0}^{\infty} \left(\dfrac{q}{\epsilon_{(0)}} \right)^l  I_{l+1}  \dfrac{(2l)!}{2^{2l} l!} \\
&= - \gamma \sqrt{ \dfrac{m}{2 \epsilon_{(0)}} } \displaystyle\int\limits_{\mathbf{x}_{(0)}}^{\mathbf{x}} \left\{\phi(\mathbf{x}') \sum\limits_{l=0}^{\infty} \left(\dfrac{q \phi(\mathbf{x}')}{\epsilon_{(0)}} \right)^l  \dfrac{(2l)!}{2^{2l} l!} \right\} \dfrac{\mathbf{v}}{|\mathbf{v}|} \cdot \rmd \mathbf{x}'
\end{array}
\]

Finally one finds that:

\[
\sum\limits_{l=0}^{\infty} \left(\dfrac{q \phi(\mathbf{x}')}{\epsilon_{(0)}}\right) \dfrac{(2l!)}{2^{2l} (l!)^2} = \left(1 - \dfrac{q \phi(\mathbf{x}')}{\epsilon_{(0)}}\right)^{-0.5}
\]

And therefore one arrives at the closed form solution:

\[
\psi = - \gamma \sqrt{\dfrac{m}{2}} \int\limits_{\mathbf{x}_{(0)}}^{\mathbf{x}}\dfrac{q\phi(\mathbf{x}')}{\sqrt{\epsilon_{(0)}(\mathbf{x},\mathbf{v}) - q \phi(\mathbf{x}')}} \dfrac{\mathbf{v}}{|\mathbf{v}|} \cdot \rmd \mathbf{x}'
\]
\end{proof}

\section{Leibniz path integral rule}
\label{app:leibniz}
One can extend the Leibniz integral rule to path integrals. Suppose that one starts with the following integral:

\[
I := \int\limits_{\mathbf{a}(\mathbf{r})}^{\mathbf{b}(\mathbf{r})} \mathbf{f}(\mathbf{r},\mathbf{t}) \cdot \rmd \mathbf{t} \equiv \sum\limits_j \int\limits_{a_j(\mathbf{r})}^{b_j(\mathbf{r})} f_j(\mathbf{r},\mathbf{t}) \rmd t_j
\]

This integral is a path integral between two points on the $D$-dimensional manifold $\mathbb{R}^D$ where the position on the manifold is given by $\mathbf{r} = \sum\limits_i r_i \mathbf{e}^{i}$. Suppose that one acts on $I$ with a vector differential operator $\hat{\mathbf{O}}$:

\[
\hat{\mathbf{O}} I = \sum\limits_j\left[ \hat{\mathbf{O}} (b_j (\mathbf{r})) f_j(\mathbf{r},\mathbf{b}) - \hat{\mathbf{O}} (a_j (\mathbf{r})) f_j(\mathbf{r},\mathbf{a}) + \int\limits_{a_j(\mathbf{r})}^{b_j(\mathbf{r})} \hat{\mathbf{O}}f_j(\mathbf{r},\mathbf{t}) \rmd t_j\right]
\]

which can be alternatively represented by using dyads:

\[
\hat{\mathbf{O}} I = \hat{\mathbf{O}} \mathbf{b}(\mathbf{r}) \cdot \mathbf{f}(\mathbf{r},\mathbf{b}) - \hat{\mathbf{O}} \mathbf{a}(\mathbf{r}) \cdot \mathbf{f}(\mathbf{r},\mathbf{a}) + \int\limits_{\mathbf{a}(\mathbf{r})}^{\mathbf{b}(\mathbf{r})} \hat{\mathbf{O}} \mathbf{f}(\mathbf{r},\mathbf{t}) \cdot \rmd \mathbf{t}
\]

\small
\bibliographystyle{unsrt}
\bibliography{main}

\end{document}